\documentclass[11pt]{article}
\usepackage{amssymb}
\usepackage{amsmath}
\usepackage{graphicx}
\usepackage{amsthm}

\newcommand{\eq}[1]{Eq.~(\ref{eq.#1})} 
\newcommand{\eqbare}[1]{(\ref{eq.#1})} 

\newcommand{\fig}[1]{Fig.~\ref{fig.#1}}

\newcommand{\sect}[1]{Sec.~\ref{sect.#1}}

\newcommand{\sectlabel}[1]{\label{sect.#1}}
\newcommand{\eqlabel}[1]{\label{eq.#1}}
\newcommand{\figlabel}[1]{\label{fig.#1}}

\newtheorem{thm}{Theorem}
\newtheorem{clm}{Claim}

\newtheorem{eg}{Example}

\newcommand{\zero}{\ensuremath{\varnothing}}  

\newcommand{\ket}[1]{\ensuremath{\left| #1 \right\rangle}} 

\newcommand{\nBits}{\ensuremath{p}} 
\newcommand{\nBitsi}{\ensuremath{{p_{\rm item}}}} 
\newcommand{\nBitsp}{\ensuremath{{p_{\rm price}}}} 
\newcommand{\nBidders}{\ensuremath{n}} 
\newcommand{\nItems}{\ensuremath{m}} 
\newcommand{\psiInit}{\ensuremath{{\psi_{\rm init}}}} 
\newcommand{\PsiInit}{\ensuremath{{\Psi_{\rm init}}}} 

\newcommand{\vBound}{\ensuremath{\bar{v}}} 

\begin{document}

\title{Quantum Auctions}
\author{Tad Hogg \\ HP Labs \\ Palo Alto, CA \and Pavithra Harsha \\ MIT \\ Cambridge, MA \and Kay-Yut Chen \\ HP Labs \\ Palo Alto, CA}

\maketitle

\begin{abstract}
We present a quantum auction protocol using superpositions to
represent bids and distributed search to identify the winner(s).
Measuring the final quantum state gives the auction outcome while
simultaneously destroying the superposition. Thus non-winning bids
are never revealed. Participants can use entanglement to arrange for
correlations among their bids, with the assurance that this
entanglement is not observable by others. The protocol is useful for
information hiding applications, such as partnership bidding with
allocative externality or concerns about revealing bidding
preferences. The protocol applies to a variety of auction types,
e.g., first or second price, and to auctions involving either a
single item or arbitrary bundles of items (i.e., combinatorial
auctions). We analyze the game-theoretical behavior of the quantum
protocol for the simple case of a sealed-bid quantum, and show how a
suitably designed adiabatic search reduces the possibilities for
bidders to game the auction. This design illustrates how incentive
rather that computational constraints affect quantum algorithm
choices.
\end{abstract}

\newpage

\section{Introduction}

Quantum information processing~\cite{chuang00} offers potential
improvements in a variety of applications. Computational
advantages~\cite{shor94,grover96} of quantum computers with many
qubits have received the most attention but are difficult to
implement physically. On the other hand, technology for
manipulating and communicating just a few qubits could be
sufficient to create new economic mechanisms by altering the
information security and strategic incentives of the underlying
game.

Examples of quantum mechanisms include the prisoner's
dilemma~\cite{eisert98,eisert00,du01,du02},
coordination~\cite{huberman03,mura03} and public goods
provisioning~\cite{chen03}. In particular, a quantum mechanism can
significantly reduce the free-rider problem without a third-party
enforcer or repeated interactions, both in theory and
practice~\cite{chen06}.

In this paper, we examine quantum mechanisms for another economic
scenario: resource allocation by auction~\cite{wilson92}. While
traditional auction mechanisms can efficiently allocate resources
in many cases, quantum auction protocols offer improvements in
preserving privacy of the losing bids and dealing with scenarios
in which bidders care about what other bidders win when multiple
items are auctioned. Specifically, using quantum superpositions to
represent bids prevents the auctioneer and other bidders from
viewing the bids during the auction without disrupting the auction
process. Furthermore, the auction result reveals nothing but the
winning bid and allocation.

The first part of the paper introduces a general quantum auction
protocol for various pricing and allocation rules, multiple unit
auctions, combinatorial auctions and partnership bids. For
simplicity, we focus on the sealed-bid first-price auction. In this
auction, each bidder has one opportunity to submit a bid. The winner
is the highest bidder, who pays the amount bid for the item. This
auction has been well studied both theoretically~\cite{wilson92} and
experimentally~\cite{cox88,chen98}, and contrasts with iterative
auctions in which bidders can incrementally increase their bids
depending on how others bid.

If the auction is not well-matched to the bidders preferences, it
can introduce perverse incentives and result in poor outcomes, such
as lost revenue for the seller or economically inefficient
allocations where items are not allocated to those who value them
most. Thus it is important to examine incentives introduced with a
proposed auction design. In particular, our auction protocol
involves quantum search, which introduces incentive issues beyond
those examined in prior quantum games~\cite{eisert00}.

A full analysis of incentive issues is complicated, even for
classical auctions. In this paper we focus on two incentive issues
arising from the quantum auction protocol. The first incentive issue
arises from the possibility of manipulating the search outcome by
altering amplitudes associated with different bids. We show how to
revise an adiabatic search method to correct this incentive problem,
thereby preserving the classical Nash equilibrium. From a quantum
algorithm perspective, this construction of the search illustrates
how incentive issues affect algorithm design, in contrast to the
more common concern with computational efficiency in quantum
information processing.

Second, the quantum search for the highest bid is probabilistic,
i.e., does not always return the highest bid. While the probability
of finding the correct answer can be made as high as one wishes by
using more iterations of the search, the small residue probability
of awarding the item to someone other than the highest bidder may
change bidding behavior. As a step toward addressing the effect of
probabilistic outcomes, we show that, with sufficient steps in the
quantum search, altering choices from those of the corresponding
deterministic auction gives at most a small improvement for that
bidder.

The paper is organized as follows. \sect{auction} describes the
quantum auction and the bidding language encoding bids in quantum
states. \sect{search} describes the quantum search method to find
the maximum bid. After these sections describing the auction
protocol, in \sect{StatBehavior} we turn to strategic issues raised
by the quantum nature of the auction beyond those in the
corresponding classical auctions. Then, in \sect{AuctionDesign} we
give a game theory analysis of some of these strategic possibilities
and describe how simple modifications of the quantum search improves
the auction outcome, in theory. \sect{Multiple} generalizes the
results to auctions of multiple items, including combinatorial
auctions. \sect{applications} describes scenarios for which the
quantum protocol offers likely economic advantages in terms of
information security and ability to compactly express complex
dependencies among items and bidders. Finally, \sect{discussion}
summarizes the quantum auction protocol and highlights a number of
remaining economic questions.

\section{Quantum Auction Protocol}\sectlabel{auction}

In our auction protocol, each bidder selects an operator that
produces the desired bid from a prespecified initial state. The
auctioneer repeatedly asks the bidders to apply their individual
operators in a distributed implementation of a quantum search to
find the winning bid. More specifically, the quantum auction
protocol for sealed-bid auctions involves the following steps:
\begin{enumerate}

\item Auctioneer announces conventional aspects of the auction: type
of auction (e.g., first or second price and any reservation prices),
the good(s) for sale, the allowed price granularity (e.g., if bids
can specify values to the penny, or only to the dollar), and the
criterion used to determine the winner(s), e.g., maximizing revenue
for the seller

\item  Auctioneer announces how quantum
states will be interpreted, i.e., as specifying a price if only one
good is for sale, or a combination of price and a set of goods if
combinations are for sale; and also announces the initial quantum
state. This state uses $\nBits$ qubits for each bidder. Auctioneer
announces the quantum search procedure.

\item Each bidder selects an operator on $\nBits$ qubits. Bidders keep their choice of operator private.

\item Auctioneer produces a set of particles implementing $\nBits$ qubits for each bidder, initializing the set to
the announced initial state.

\item Auctioneer and bidders perform a distributed search for the
winner

\end{enumerate}

\fig{schematic} illustrates this procedure for two bidders and
repeating the steps of the search twice. Realistic search involves a
larger number of steps. In contrast with other quantum games, e.g.,
public goods, that involve just one round of interaction, the search
required to identify the winners involves multiple rounds of
interaction among the participants. The required number of
iterations depends on the search method. In practice, the auctioneer
could pick the number of iterations based on prior experience with
similar auctions, or from simulating several test cases using
valuations randomly drawn from a plausible distribution of values
for the auction items. Alternatively, the auctioneer could repeat
the procedure several times (possibly with steps from each
repetition interleaved in a random order) and use the best result
from these repetitions.

\begin{figure}
\begin{center}
\includegraphics[scale=0.6]{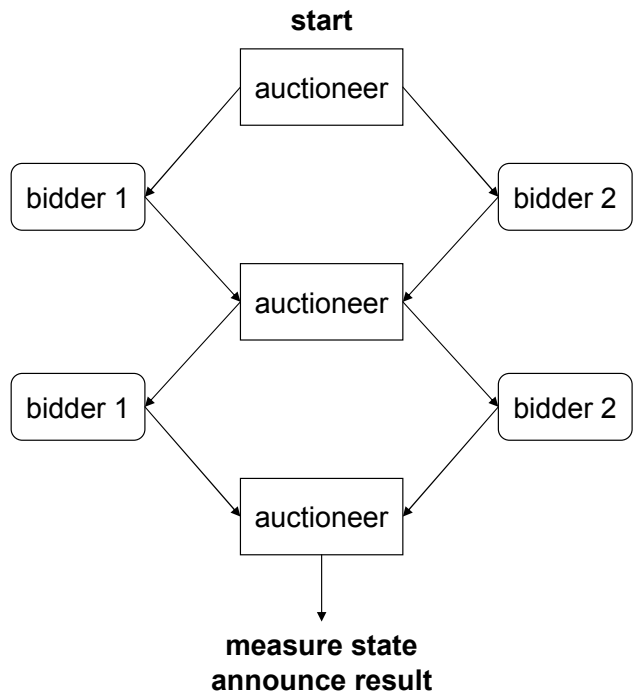}
\end{center}
\caption{Schematic diagram of distributed search procedure,
showing repeated interactions between auctioneer and bidders, in
this case two bidders and two steps of the distributed search.
}\figlabel{schematic}
\end{figure}

This auction protocol uses a distributed search so bidders'
operator choices remain private. Specifically, the search
operation requiring input from the bidders is applied locally by
each bidder, giving the overall operator
\begin{equation}\eqlabel{U}
U = U_1 \otimes U_2 \otimes \ldots \otimes U_\nBidders
\end{equation}
where $\nBidders$ is the number of bidders and $U_i$ the operator
of bidder $i$.

\begin{table}
\begin{center}
\begin{tabular}{l|l}
number of bidders   & \nBidders  \\
\hline number of items in auction  & \nItems  \\
\hline number of qubits per bidder & \nBits \\
\hline state of qubits for bidder $j$ & $\psi_j$ \\
\hline state of all qubits & $\Psi = \psi_1 \otimes \ldots \otimes \psi_\nBidders$ \\
\end{tabular}
\end{center}
\caption{Notation for the quantum auction.}
\end{table}

\section{Quantum Auction Implementation}\sectlabel{search}

A quantum auction requires finding the winning bid and corresponding
bidder. This procedure has two components: the interpretation of the
qubits as bids, and the search procedure to find the winner. The
following two subsections discuss these components in the context of
a single-item auction. \sect{Multiple} generalizes this discussion
to multiple items.

\subsection{Creation and interpretation of quantum bids}\sectlabel{bids}

We define a {\em bid} as the amount a bidder indicates he is
willing to pay for the item. An {\em allocation} is a list of
bids, one from each bidder. The quantum auction protocol
manipulates superpositions of allocations. We use an {\em
allocation rule} to indicate how allocations specify a winner and
amount paid.

\begin{eg}
Consider an auction of one item with three bidders, willing to pay
\$1, \$3 and \$10 for the item, respectively. We represent these
bids as $\ket{\$1}$, $\ket{\$3}$ and $\ket{\$10}$, and the
corresponding allocation as the product of these states, i.e.,
$\ket{\$1,\$3,\$10}$ with the ordering in the allocation
understood to correspond to the bidders. A simple allocation rule
selects the highest bidder as the winner, who pays the high bid.
In this example, this rule results in the third bidder winning,
and paying \$10 for the item.
\end{eg}

Each bidder gets \nBits\ qubits and can only operate on those bits.
Thus each bidder has $2^\nBits$ possible bid values, and can create
superpositions of these values. A superposition of bids specifies
set of distinct bids, with at most one allowed to win. The
amplitudes of the superposition affect the likelihood of various
outcomes for the auction. For a single-item auction, a bidder will
typically have only one bid. As discussed below, more complicated
superpositions are useful for information hiding. Specifically,
bidder $j$ selects an operator $U_{j}$ on \nBits\ qubits to apply to
the initial state for that bidder's qubits \psiInit\ specified by
the auctioneer. The resulting state, $\psi_{j} = U_{j} \psiInit$, is
a superposition of bids, each of the form $\ket{b_{i}^{(j)}}$ where
$b_{i}^{(j)}$ is bidder $j$'s bid for the item. The subscript $i$
indicates one of the possible bids that can be specified with
\nBits\ qubits according to the announced interpretation of the
bits.

We define the subspace used by bidder $j$ as the set of states
spanned by the basis eigenvectors in $\psi_{j}$. Only these basis
vectors appear in allocations relevant for the search. As bidders
apply their operators during the search, the superposition of
allocations remains within the subspace of each bidder. In this
case, where each bidder applies an operator only to their own
qubits, the superposition of allocations is always a factored form,
i.e., $\Psi = \psi_{1} \otimes \ldots \otimes \psi_{n}$. More
generally, groups of bidders could operate jointly on their qubits,
entangling their bids in the allocations as discussed in
\sect{applications}.

To exploit information hiding properties of superpositions, the
state revealed at the end of the search should specify only the
bidder who wins the item and the corresponding bid. To achieve this,
instead of a direct representation of bids, we interpret bids formed
from the \nBits\ qubits available to a bidder as containing a
special null value, \zero, indicating a bid for nothing. This null
bid has additional benefits in multiple item settings, as discussed
in \sect{Multiple} and \sect{applications}.

\begin{eg}
Consider bidder $j$ with two qubits and the initial state
$\psiInit = \ket{00}$ corresponding to the vector $(1, 0, 0, 0)$,
which is interpreted as the null bid. The other bid states are
\ket{01}, \ket{10} and \ket{11} corresponding to vectors
(0,1,0,0), (0,0,1,0) and (0,0,0,1). These three states are
interpreted as three bid values in some preannounced way, e.g.,
\$1, \$2 and \$3, respectively.

The operator
\begin{equation}
U_j = \frac{1}{\sqrt{2}}
\begin{pmatrix}
    1 & 0 & 1 & 0 \cr
    0 & 1 & 0 & 1 \cr
    1 & 0 & -1 & 0 \cr
    0 & 1 & 0 & -1 \cr
\end{pmatrix}
\end{equation}
gives the initial state $\psi_{j} = U_{j} \psiInit$ as
$(\ket{00}+\ket{10})/\sqrt{2}$ and specifies the search subspace
whose basis is the first and third columns of $U_j$ in this example.
Thus the possible allocations involve only $\ket{00}$ and $\ket{10}$
for this bidder, corresponding to the null bid and a bid of \$2,
respectively.
\end{eg}

In the presence of a null bid, we consider an allocation to be a
{\em feasible} if it contains exactly one bid not equal to \zero.
The corresponding allocation rule assigns no winner to infeasible
allocations and, for feasible allocations, the winner is the single
bidder in the allocation whose bid is not \zero, and he pays the
amount bid. This allocation rule corresponds to a first-price
single-item auction, except there can be no winner, analogous to the
situation in auctions with a reservation price when no bidder
exceeds that price.

\subsection{Distributed Search}\sectlabel{adiabatic}

The auctioneer must find the best state according to an announced
criterion, e.g., maximum revenue. Specifically, the auctioneer has a
evaluation function $F$ assigning a quality value to each
allocation. The function $F$ assigns a lower value to infeasible
allocations than to any feasible one. An example is $F$ equal to the
revenue produced by the allocation (if feasible) and otherwise is
$-1$.

The auctioneer uses quantum search to find the allocation in the
subspace selected by the bidders giving the maximum value for $F$
(e.g., a feasible allocation giving the most revenue to the
auctioneer). This could be done via repeated uses of a
decision-problem quantum search~\cite{grover96,boyer96} as a
subroutine within a search for the minimum threshold value of $F$
giving a solution to the decision problem, e.g., with a classical
binary search on threshold values or using results of prior
iterations of the decision problem~\cite{durr96}. Alternatively, we
could use a method giving the maximum value directly (e.g.,
adiabatic~\cite{farhi01} if run for a sufficiently long time or
heuristic methods~\cite{hogg00,hogg03} based on some prior knowledge
of the distribution of bidders values). For definiteness, we focus
on the adiabatic method.

The adiabatic search is conventionally described as searching for
the minimum {\em cost} state. We use this convention by defining a
state's cost to be the negative of the evaluation function $F$. The
adiabatic search procedure, if run sufficiently slowly, changes the
initial superposition into a final superposition in such a way that
the amplitude in each initial eigenstate maps to the same amplitude
in the corresponding final eigenstate, up to a phase factor (for
nondegenerate eigenstates). We refer to this mapping of initial to
final eigenstates as a {\em perfect search}. In practice, with a
finite time for the search, there will be some transfer of amplitude
among the eigenstates so the search will not be perfect in the sense
defined here. Instead the auction outcome is probabilistic: the
auction will not always produce the best outcome when starting from
the ground state. For example, an auction intending to find the
highest bid could sometimes produce the second highest bid instead.
Conventionally, the search operations are chosen so the uniform
superposition is the lowest cost initial eigenstate. In our case,
bidders are free to choose their operators and need not create
uniform superpositions.

A discrete implementation of adiabatic search consists of the
following steps:
\begin{itemize}
\item The auctioneer selects a number of search steps $S$ and
parameter $\Delta$. These need not be announced to the bidders.

\item The auctioneer initializes
the state of all $\nBidders \nBits$ qubits to $\PsiInit = \psiInit
\otimes \ldots \otimes \psiInit = \ket{0,\ldots,0}$, with
$\nBidders$ factors of $\psiInit$ in the product, and
$\psiInit=\ket{0}$ is the initial state for the $\nBits$ qubits
for a single bidder.

\item The auctioneer sends these initialized qubits to the bidders
who use their individual operators and then return the qubits to
the auctioneer, jointly creating the state
\begin{equation}\eqlabel{search init}
\Psi_0 = U \PsiInit
\end{equation}

\item For $s=1,\ldots,S$, the auctioneer and bidders update the state to
\begin{equation}\eqlabel{search update}
\Psi_s = U D(f) U^\dagger P(f) \Psi_{s-1}
\end{equation}
with $f = s/S$ the fraction of steps completed. The bid operator
$U$ and its adjoint $U^\dagger$ are performed by sending bits to
the bidders as described in \sect{auction}. The diagonal matrices
$D(f)$ and $P(f)$ are described below.

\item The auctioneer measures the state $\Psi_S$, resulting in specific values for all the bits, from which the winner and prices are determined
by the allocation rule described in \sect{bids}.

\end{itemize}

The diagonal matrix $P(f)$ adjusts the phases of the amplitudes
according to the cost associated with each allocation. In
particular, using the cost $c(x) = -F(x)$ for allocation
$\ket{x}$, we have
\begin{equation}\eqlabel{search Hc}
P_{x x}(f) = \exp \left( -i f c(x) \Delta \right)
\end{equation}
Similarly, the diagonal matrix $D(f)$ adjusts amplitude phases as
defined by a function $d(x)$:
\begin{equation}\eqlabel{search H0}
D_{x x}(f) = \exp \left( -i (1-f) d(x) \Delta \right)
\end{equation}
The key property of $d(x)$ is assigning the smallest value, e.g., 0,
to $\ket{0}$, thereby making the first column of $U$ the ground
state eigenvector. Aside from this key property, the choice of
$d(x)$ is somewhat arbitrary. The conventional choice in the
adiabatic method uses the Hamming weight of the state, i.e., $d(x)$
equal to the number of 1 bits in the binary representation of $x$.
However, as described in \sect{AuctionDesign}, other choices for
$d(x)$ can improve the incentive properties of the auction.

The discrete-step implementation of the continuous adiabatic
method~\cite{farhi01} involves the limits $\Delta \rightarrow 0$
and $S \Delta \rightarrow \infty$, in which case the final state
$\psi_S$ has high probability to be the lowest cost state. In
practice, this outcome can often be achieved with considerably
fewer steps using a fixed value of $\Delta$, corresponding to a
discrete version of the adiabatic method~\cite{hogg03}.

\section{Strategies with Quantum Operators}\sectlabel{StatBehavior}

Ideally, an auction achieves the economic objective of its design
(e.g. maximum revenue for the seller). In practice, an auction
design may not provide incentives for participants to behave so as
to achieve this objective. Usually auction designs are examined
under the assumption of self-interested rational participants. In
conventional auctions, strategic issues include misrepresentation
of the true value, collusion among bidders and false name bidding
(where a single bidder submits bids under several aliases). Some
of these issues can be addressed with suitable auction rules,
e.g., second price auctions encourage truthful reporting of
values. Developing suitable designs of classical auctions in a
wide range of economic contexts remains a challenging
problem~\cite{wilson92}.

Quantum auctions raise strategic issues beyond those of classical
auctions. In our case, every step of the adiabatic search requires
each bidder to perform an operation on their qubits. Ideally, the
bidder should use the same operator $U$ for creating $\psiInit$ as
in every step of the search in \eq{search update}. In addition,
bidders should include the null bid in their subspaces. In the
classical first-price sealed-bid auction, the bidder makes one
choice: the amount to bid. In our quantum setting, this choice
amounts to selecting the subspace to use with the quantum search.
The remaining freedom to select $U$, and possibly a different $U$
for each step in the search, are additional choices provided by
the quantum auction.

Bidders may be tempted to exploit the flexibility of choosing
operators in several general ways. First, they could use a
subspace not including the null bid. Second, they could use a
different operator for creating $\psiInit$ than they use in the
rest of the search, thereby producing an altered initial amplitude
that is not the ground state eigenvector. Third they could change
operators during the search. If any such changes give significant
probability for low bids to win, bidders would be tempted to make
such changes and include a low bid in their subspace, hoping to
profit significantly by winning the auction with a low bid.

The remainder of this section describes some strategic issues
unique to quantum auctions and possible solutions. We further
discuss a game theory analysis of some of these issues in
\sect{AuctionDesign}.

\subsection{Selecting the Subspace}\sectlabel{subspace choice}

The use of the null bid in our protocol raises the strategic issue
illustrated in the following example:

\begin{eg} Consider an auction of a single item with two bidders
Alice and Bob. Using operators producing uniform amplitudes for
the sake of illustration, they ought to apply operators that
create
\begin{center}
$\frac{1}{\sqrt{2}}(\ket{\zero} + \ket{b_{A}})$ and
$\frac{1}{\sqrt{2}}(\ket{\zero} + \ket{b_{B}})$
\end{center}
respectively, where $b_A$ and $b_B$ are their desired bids. The
initial superposition for all the qubits is the product of these
individual superpositions, i.e., $\Psi_0$ is
\begin{displaymath}
\frac{1}{2}(\ket{\zero,\zero} + \ket{b_{A},\zero}+
\ket{\zero,b_{B}} +\ket{b_{A},b_{B}} )
\end{displaymath}
If bidders use these same operators during the search, the search
algorithm finds the highest revenue allocation, i.e., giving the
item to the highest bidder. Suppose instead Bob picks an operator
with a one-dimensional subspace, producing an initial state
$\ket{b_{B}}$ rather than including $\zero$. The product
superposition is then
\begin{displaymath}
\frac{1}{\sqrt{2}}(\ket{\zero,b_{B}} +\ket{b_{A},b_{B}} )
\end{displaymath}
Since the search remains in this subspace and the second
allocation is infeasible, the search will return
$\ket{\zero,b_{B}}$ no matter what Alice bids. Thus Bob always
wins the item, and can win using the lowest possible bid.
\end{eg}

This example shows bidders have an incentive to exclude the null set
from their subspace. If all bidders make this choice, there will be
no feasible allocations in the joint subspace and the auction will
always give no winner. For auctions with more than two bidders,
selecting subspaces excluding $\zero$ is a weak Nash Equilibrium for
the quantum auction because any other choice by a single bidder
still results in no feasible allocations.

\subsection{Altering Initial Amplitudes}\sectlabel{initial
amplitudes}

Strategic choices for bidders also arise from the search procedure
itself, even when using the correct subspace consisting of $\zero$
and the desired bid. In particular, the probabilistic outcome of
the search means the optimal bid according to the auction
criterion (e.g., highest revenue) will not always win. For the
adiabatic search method, bidders could try to arrange for
especially tiny eigenvalue gaps between the state corresponding to
the best outcome and another state allowing them to win with a low
bid. A sufficiently small gap could make the number of steps the
auctioneer selects insufficient to give the optimal state with
high probability and instead give a significant chance of
producing the more favorable outcome. However, because the
eigenvalues are a complicated function of the operators of all
bidders, and individual bidders do not know the choices made by
others, it will be difficult for a bidder to determine how to make
such especially small gaps and do so in a way that gives a
favorable outcome. Nevertheless, even fairly small probabilities
for not finding the optimal state could alter the strategic
behavior of the bidders.

\begin{figure}[t]
\begin{center}
\includegraphics[scale=0.8]{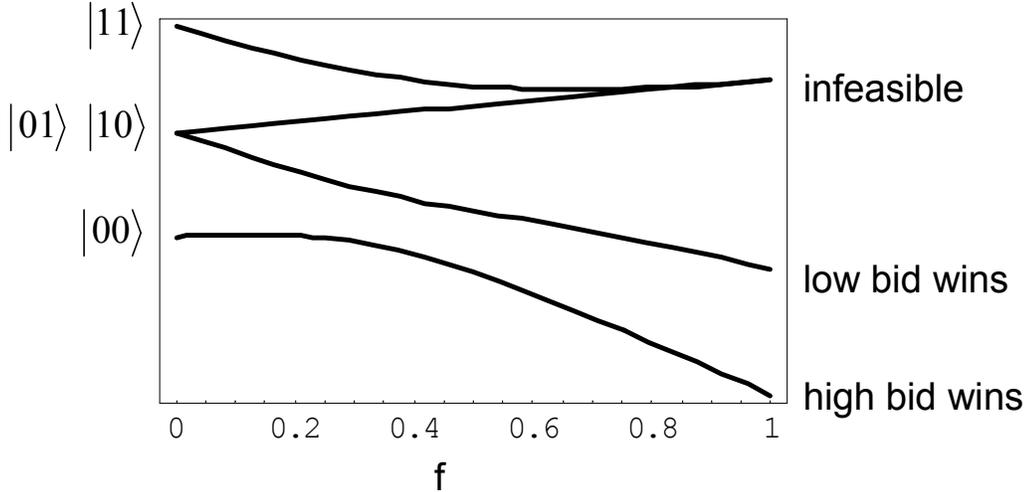} 
\caption{Correspondence between initial basis and the possible
allocations for a single item auction with two bidders in the
standard adiabatic search. During the search, as $f$ increases from
0 to 1, the eigenvalues of the four states change as shown
schematically in the figure. The states for $f=0$ correspond to both
bidders starting with the ground state, $\ket{00}$, the two states
obtained if one of the bidders starts with a different
superposition, $\ket{01}$ and $\ket{10}$ (``single-bidder deviation
states''), and the state of both bidders starting with different
superpositions, $\ket{11}$ (``2-bidder deviation
state'').}\figlabel{standard}
\end{center}
\end{figure}

A more direct way a bidder can arrange for a low bid to win is by
altering the initial state of the adiabatic search to start not in
the ground state but in an eigenvector corresponding to one of the
first few eigenvalues above the ground state. The adiabatic search
takes such eigenvectors, with high probability, to an outcome in
which a bid lower than the highest wins. While a single bidder
cannot create an arbitrary initial condition, one bidder can
ensure that it is not the ground state. For example, a bidder
could chose an operator that gives a nonuniform amplitude for the
initial state, in particular $(\ket{\zero} -
\ket{b_{A}})/\sqrt{2}$, while using the uniform state
$(\ket{\zero} + \ket{b_{A}})/\sqrt{2}$ as the ground state through
the remainder of the search in \eq{search update}. This can result
in significant probability for a low bid to win, and so a bidder
is tempted to deviate from the nominal operator choice.

\fig{standard} illustrates this behavior. Instead of starting in
the ground state, the bidder's choice gives the initial state as a
linear combination of the ground state and the single-deviation
state for that bidder, denoted as $\ket{01}$ or $\ket{10}$ for the
two bidders in \fig{standard}. Here a ``single deviation'' state
is one that a single bidder can create, i.e., by operating on just
the qubits available to that bidder. The adiabatic search splits
the degeneracy, thereby giving some probability for the lowest bid
to win and some probability for an infeasible allocation.

More generally, bidder $i$ uses this strategy by selecting two
different operators $U_i^{\rm init}$ and $U_i$ to use for forming
the initial state and during the search, respectively. These
choices result in different joint operators, in \eq{U}, used in
\eq{search init} and \eqbare{search update}.

As with selecting a subspace without $\zero$, if many or all
bidders make this choice, the initial state will have significant
amplitude in eigenvectors corresponding to large eigenvalues,
which produce infeasible outcomes and hence a high probability for
no winner. Thus with standard adiabatic search, if everyone uses
the same operator for both initialization and search, then each
bidder is tempted to use a different initialization operator and
bid low, gaining a chance to win with a low bid. However, if
multiple bidders attempt this, the outcome will most likely be an
infeasible state, with no winner.

We can address this problem by reordering the eigenvalues given by
the $d(x)$ function in \eq{search H0} so that any change in
initial operator by a single bidder increases probability of
infeasible allocation but not the probability of any feasible
allocation with a bid lower than the highest bid. This is possible
because bidders only have access to their own bits, so can only
form initial superpositions from a limited set of basis vectors.
\fig{permuted} illustrates the resulting situation. We give an
analysis of this approach in \sect{behavior}.

\begin{figure}[t]
\begin{center}
\includegraphics[scale=0.8]{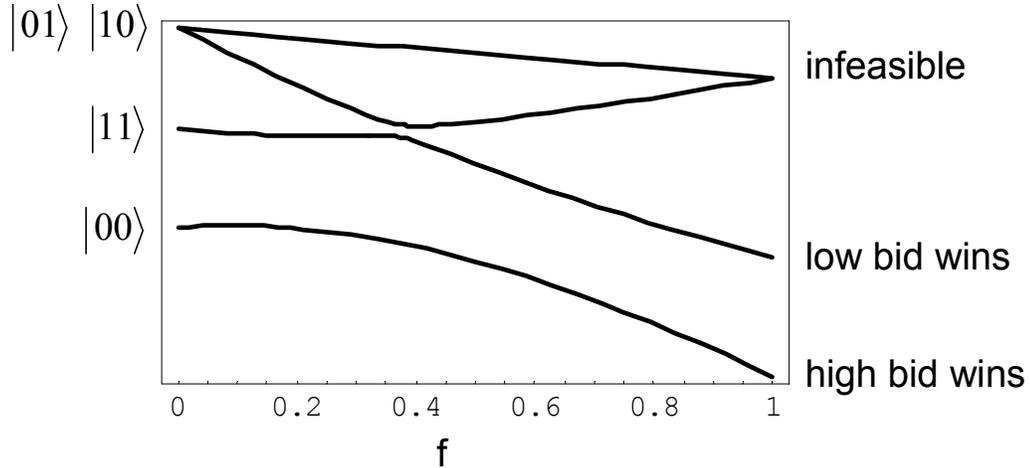} 
\caption{Correspondence between the initial basis and the possible
allocations for a single item auction with two bidders in the
search with permuted initial eigenvalues.}\figlabel{permuted}
\end{center}
\end{figure}

\subsection{Changing Operator During Search}\sectlabel{change
during search}

The distributed search of \eq{search update} has each bidder using
the same operator for every step of the search. Thus bidders may
gain some advantage by altering their operator during the steps of
the search. Gradually changing the operator during the search
amounts to a different path from initial to final Hamiltonian
during the adiabatic search. Thus, provided the auctioneer uses
enough steps, such changes will have at most a minor effect on the
outcome probabilities unless the bidder can arrange for
particularly small eigenvalue gaps among favorable states. Such
arrangement is difficult, particularly since the bidder does not
know the choices of other bidders and the auctioneer could treat
the bits from the bidders in an arbitrary, unannounced order.

More significant changes in outcome is possible with sudden, large
changes in the operator during search. Since the use of bidders
operators gradually decreases during the search (i.e., $D_{x
x}(f)$ given in \eq{search H0} approaches the identity operator as
$f$ approaches 1), the most problematic situation is for an abrupt
change in operator at the beginning of the search. After such a
change, the adiabatic search continues its gradual change of
states, but now instead of starting in the ground state, it will
instead have a linear combination of various states obtained by
mapping the original basis onto the basis after the change.

\section{Quantum Auction Design}\sectlabel{AuctionDesign}

In this section, we focus on mechanism design to reduce incentive
issues arising from the quantum aspects of the auction. We analyze
incentive issues with the Nash equilibrium (NE) concept commonly
used to evaluate auctions~\cite{wilson92}. A given set of behaviors
for the bidders is an equilibrium if no single bidder can gain an
advantage (i.e., higher expected payoff) by switching to another
behavior. Specifically, \sect{null set} describes an approach to
encouraging bidders to include the null set in their bids. In
\sect{behavior} we show that using the ground state eigenvector is a
NE provided bidders do not change the operators during the search.
\sect{test change during search} then discusses how the auctioneer
can discourage bidders from changing operators. \sect{auctioneer}
describes how the auction can be made symmetric across the different
bidders. We focus on single-item auctions in this section, but the
ideas extend to quantum combinatorial auctions, as described in
\sect{Multiple}.

\subsection{Checking for the Null Set}\sectlabel{null set}

One approach to the incentive to exclude the null set, described in
\sect{subspace choice}, is for the auctioneer to perform a second
search: for the allocation with the most \zero\ values. This search
uses the same distributed protocol of \eq{search update} but with
separate qubits and a different cost function to define $P(f)$,
i.e., setting $c(x)$ to the number of non-\zero\ values in the
allocation $x$. Interleaving the additional search in a random order
within the steps of the search for the winning bid prevents bidders
from knowing which search a given step belongs to. So bidders could
not consistently select different operators for the two searches.

If all bidders include $\zero$ in their selected subspace, this
additional search returns $\ket{\zero,\zero,\ldots}$. Any bidder
found not to have included $\zero$ could be excluded from winning
the auction. At this point the auctioneer could either announce
there is no winner, or restart the auction for the remaining bidders
without announcing this restart. The adiabatic search has a small
but nonzero probability of returning the wrong result, which would
then incorrectly conclude some bidder did not include \zero. As long
as the probability of such errors is smaller than the error
probability of the search for the winner, these errors should not
greatly affect the incentive structure of the mechanism.
Alternatively, the auctioneer could use a search completing with
probability one in a finite number of steps, i.e., with different
choices of $D$ and $P$ in \eq{search update}, the auctioneer could
implement Grover's algorithm~\cite{grover96} to search for the
allocation $\ket{\zero,\zero,\ldots}$ in the joint subspace of the
bidders. Since the auctioneer does not know the size of the
subspaces selected by the bidders, the auctioneer would need to try
various numbers of steps~\cite{boyer96} before concluding
$\ket{\zero,\zero,\ldots}$ is not in the selected subspaces. Unlike
the adiabatic search, failure would only indicate \emph{some} bidder
had not included \zero, but not which one. Thus the auctioneer's
only alternative in this case is to announce the auction has no
winner.

While this approach removes the immediate benefit of not including
the null bid, its affect on broader strategic issues in the full
auction is an open question.

\subsection{The First-Price Sealed-Bid
Auction}\sectlabel{behavior}

In this section we examine the incentive structure of the auction
with permuted eigenstates described in \sect{initial amplitudes}. We
first review how game theory applies to auctions. We then consider
the quantum auction when the search runs long enough to give
successful completion almost always (``perfect search''). Finally,
we consider the more realistic case of search with small, but not
negligible, probability for non-optimal outcomes.

\subsubsection{A Game Theory Approach to Auctions}\sectlabel{FPA}

Game theory is a common approach to evaluating
auctions~\cite{milgrom82,wilson92}. Consider $n$ people bidding for
an item, with person $i$ having value $v_i$ for the item. Unlike
discrete choice games, such as the prisoner's dilemma, a strategy
for a private value auction involves a bidding \emph{function}
$b(v)$, mapping a bidder's value to a corresponding bid. Theoretical
analysis of auctions usually involves identifying a NE strategy, if
any. This is a strategy for all players such that no bidder gains by
changing this strategy given everyone else is using it. This focus
on possible changes by a single bidder assumes bidders do not
collude.

A primary issue for auction behavior is how much participants know
about other bidders' values. Such knowledge can affect the choice of
bid. The most popular model of such knowledge is independent private
values, where the $v_i$ are independently drawn from the same
distribution. Each bidder knows his own value, but not the values of
other bidders. However, the distribution from which values come is
common knowledge, i.e., known to all bidders, each bidder knows the
others know this fact, and so on. A final ingredient for the
analysis is an assumption of bidders' goals. For illustration, we
use the common assumption that bidders are risk neutral expected
utility maximizers, and within the context of the auction, utility
is proportional to profit.

We illustrate this approach for a first-price sealed-bid auction, in
which each bidder submits a single bid without seeing any of the
other bids. This corresponds to the auction scenario considered in
this paper. The bidder with the highest bid gets the item and pays
the amount of his bid. Thus if bidder $i$ bids $b_i$, his profit is
$v_i-b_i$ if he wins the auction and zero otherwise. To avoid
possibly losing money, bidders should ensure $b_i \le v_i$, and bids
are required to be nonnegative.

In the symmetric case where bidders' values all come from the same
distribution, a NE is a bidding function $b(v)$. A bidder's expected
payoff is $(v-b(v))P(b)$ where $v$ is his value, $b$ is his bidding
function and $P(b)$ is the probability of winning if he is using
$b(v)$ (which is also the function others use in equilibrium). Let
$F$ be the cumulative distribution of values, i.e., probability a
value is at most $v$, and $n$ be the number of bidders. The
equilibrium condition leads to a differential equation satisfied by
$b(v)$~\cite{wilson92}. As a simple example, when $v$ is uniformly
distributed between 0 and 1, $F(v) = v$ and the NE is $b(v) = (n-1)
v / n$. Thus, in the equilibrium strategy, a bidder bids somewhat
less than his value and the bid gets closer to the value when there
is more competition, i.e., larger $n$.

If bidders have differing value distributions, a NE involves a set
of bidding functions, $\{b_i(v)\}$. An auction may have multiple
equilibria.

\subsubsection{Behavior with Perfect Search}

With perfect search and non-colluding bidders, if bidders use the
same operators for every step of the search, including
initialization, and pick a subspace with the null bid then the
adiabatic search described in \sect{adiabatic} finds the highest
revenue state. We now show that the auctioneer can choose
eigenvalues for the search so that  bidders have no incentive to
create an initial state different from the ground state. This
choice corresponds to the auctioneer selecting an appropriate
function $d(x)$ in \eq{search H0}.

Suppose bidder $i$ uses operator $U_i$, giving the overall
operator $U$ with \eq{U}. Suppose all bidders except bidder 1 use
the same operator to create the initial state as they use for the
subsequent search. But bidder 1 uses two operators: $U_1^{\rm
init}$ to form the initial state and $U_1$ for the search. Thus
the initial state produced by bidder 1, $\psi_{1} = U_{1}^{\rm
init} \psiInit$, i.e., the first column of $U_{1}^{\rm init}$, is
not necessarily equal to the first column of $U_{1}$ that bidder 1
uses for the subsequent search. Instead, $\psi_1$ may have
contributions from all columns of $U_{1}$, i.e.,
\begin{equation}
\psi_1 = \sum_{i=0}^{2^\nBits-1} \alpha_i \ket{i}
\end{equation}
where $\ket{i}$ corresponds to column $i$, ranging from $0$ to
$2^\nBits-1$, of $U_1$. Combining with the initial state of all
other bidders, \eq{search init} gives $\Psi_0 = \sum_i \alpha_i
\ket{i,0,\ldots,0}$, instead of the initial ground state
$\ket{0,0,\ldots,0}$.

Significantly, because a bidder can only operate on the $\nBits$
qubits from the auctioneer and not on any of the qubits sent to
other bidders, a single bidder can only create a limited set of
``single-deviation'' initial states. In the case of bidder 1,
these states all have the form $\ket{i,0,\ldots,0}$. Similarly, if
bidder $j$ is the one using different initial and search
operators, the states all have the form
$\ket{\ldots,0,i,0,\ldots}$, where only the $j^{th}$ position can
be nonzero. Thus, among the $2^{\nBidders \nBits}$ basis states in
the full search space, aside from the correct ground state, only
$\nBidders (2^\nBits - 1)$ are possible states some single bidder
can create when all other bidders use the same operator for
initialization and search.

More generally, $k$ bidders can create superpositions of
$(2^\nBits-1)^k$ basis states in which none of them use the ground
state initially, by selecting different operators for initialization
and search. Thus there are
\begin{equation}\eqlabel{k deviation states}
\binom{\nBidders}{k} (2^\nBits-1)^k 
\end{equation}
$k$-deviation states that some set of $k$ bidders can create, while
the other $n-k$ bidders use the ground state.

Our formulation has $\nBidders(2^\nBits-1)$ feasible allocations,
i.e., situations in which exactly one of the bidders has a
non-\zero\ bid while all other bidders have \zero. To see this, each
of the $\nBidders$ bidders could have the non-\zero\ bid, and this
bid could have any of $2^\nBits-1$ values (since the remaining value
for the bidder's bits represents \zero). The remaining $\nBidders-1$
bidders have only one choice each, i.e., \zero.

Suppose the auctioneer selects $d(x)$ such that
$d(\ket{0,\ldots,0})=0$ is the lowest eigenvalue and $d(x)$ for all
single-deviation states $x$ is the largest value, with intermediate
values for all other states. Provided the number of infeasible
allocations is at least equal to the number of single-deviation
states, a perfect search will then map every single-deviation state
to an infeasible allocation, resulting in no winner for the auction.
This condition amounts to
\begin{equation}\eqlabel{condition}
2^{\nBidders \nBits} - \nBidders(2^\nBits-1) \geq \nBidders
(2^\nBits - 1)
\end{equation}
The following claim shows that \eq{condition} always is true in an
auction scenario.
\begin{clm}
\eq{condition} is true for all integers $\nBidders, \nBits \geq 1$
\end{clm}
\begin{proof}
When $\nBits = 1$, \eq{condition} reduces to $2^{\nBidders-1} \geq
\nBidders$, which is true for all $\nBidders \geq 1$.

We prove a stronger condition for $\nBits \geq 2$, namely there are
enough infeasible states to handle up to $\nBidders-1$ bidders
deviating. Using \eq{k deviation states}, this stronger condition is
\begin{equation}\eqlabel{condition with collusion}
2^{\nBidders \nBits} - \nBidders(2^\nBits-1) \geq \sum_{k=1}^{n-1}
\binom{\nBidders}{k} (2^\nBits - 1) =2^{\nBidders \nBits} - 1 -
(2^\nBits - 1)^\nBidders
\end{equation}
with the $k=1$ term in the sum corresponding to the right-hand side
of \eq{condition}. Writing $x \equiv 2^\nBits -1$, \eq{condition
with collusion} becomes $f(x,\nBidders) \equiv x^\nBidders -
\nBidders x +1 \geq 0$.

Since $\nBits \geq 2$, we have $x \geq 3$. For this range of $x$ and
for $\nBidders \geq 1$, $f(x,\nBidders)$ is monotonically increasing
in both arguments. To see $f$ is monotonic for $x$, the derivative
of $f(x,\nBidders)$ with respect to $x$ is
$\nBidders(x^{\nBidders-1}-1)$ which is nonnegative since $\nBidders
\geq 1$ and $x > 1$. Similarly, the derivative with respect to
$\nBidders$ is $x(x^{\nBidders-1} \ln(x) - 1)$ which is at least
$3(\ln(3) - 1)> 0$ since $\nBidders \geq 1$ and $x \geq 3$. Thus for
the relevant range of $\nBidders$ and $x$, $f(x,\nBidders) \geq
f(3,1) = 1$ so \eq{condition with collusion} is true for all
$\nBidders \geq 1$ and $\nBits \geq 2$.

Combining these cases for $\nBits=1$ and $\nBits \geq 2$ establishes
the claim.
\end{proof}

Using this claim, we demonstrate the permuted eigenvalue choices
remove incentives to alter the initial amplitudes:

\begin{thm}\label{thm.NE}
If (a) auctioneer chooses eigenvalues as described above, (b)
$\{b^*_i(v)\}_{i=1}^\nBidders$ is an equilibrium for the first-price
classical auction, and (c) bidders include the null set as part of
their bids and use the same operator in each step in the search
except, possibly, for the initial state, then the strategy of using
bidding functions $\{b^*_i(v)\}_{i=1}^\nBidders$ and the same
operator for their initial state as they use in the search is a NE
for corresponding quantum auction.
\end{thm}
\begin{proof}
Without loss of generality, suppose only bidder 1 deviates and all
the other bidders use $\{b^*_i(v)\}_{i=2}^\nBidders$ and the same
operator for initialization and search. Then, as described above,
the initial state $\Psi_0$ is $\sum_i \alpha_i \ket{i,0,\ldots,0}$
for some choice of amplitudes $\alpha_i$, with $i$ ranging from $0$
to $2^\nBits - 1$.

A perfect adiabatic search maps each of these states to a
corresponding allocation. In particular, with
$d(\ket{0,\ldots,0})$ having the smallest value of the function
$d(x)$, the lowest cost allocation is produced with probability
$|\alpha_0|^2$. This allocation corresponds to the highest bid
winning. Moreover, each $\ket{i,0,\ldots,0}$ with $i \neq 0$ has
the largest value of $d(x)$, and so, because of \eq{condition},
maps to an infeasible allocation, giving no winner and hence no
value to bidder 1.

Hence the expected value for bidder 1 is $|\alpha_0|^2 V$ where $V$
is the value of the expected profit of the corresponding classical
auction to bidder 1. Since $|\alpha_0|^2 V \leq V$, bidder 1 cannot
gain from such a deviation.

Furthermore, there is no gain from deviating from the bidding
function $b^*_1(v)$ since it will only decrease $V$, because, by
assumption, $\{b^*_i(v)\}_{i=1}^\nBidders$ is a NE for the
corresponding classical auction.

Because of \eq{condition}, this discussion applies to deviations by
\emph{any} single bidder, not just bidder 1. Thus, using bidding
function $\{b^*_i(v)\}_{i=1}^\nBidders$ and using the same operator
for their initial state as they use in the search is a NE.

\end{proof}

The stronger condition, \eq{condition with collusion}, shows that
the number of infeasible states is enough to give no winner for any
choice of initial amplitudes that up to $\nBidders-1$ bidders can
produce, provided $\nBits \geq 2$. Thus if an auctioneer implements
a collusion-proof classical auction with the quantum protocol and
assigns infeasible states as described then the resulting quantum
auction is collusion-proof up to $\nBidders-1$ bidders for initial
amplitude deviations.

The choice for $d(x)$ satisfying the above requirements is not
unique. As one example, let $x$ be the state index in the full
search space, running from 0 to $2^{\nBidders \nBits}-1$. Consider
$x$ as written as a series of $\nBidders$ base-$2^\nBits$ numbers,
$\ket{x_1,x_2,\ldots,x_n}$. Define
\begin{equation}\eqlabel{d(x) example}
d(x) = -r(x) \pmod{\nBidders+1}
\end{equation}
where $r(x)$ is number of nonzero values among $x_1,x_2,...,x_n$.
The mod operation gives all $d(x)$ values in the range 0 to
\nBidders. For the initial ground state, $x=\ket{0,\ldots,0}$,
$r(x)=0$ so $d(x)=0$, and this is the smallest possible value.
Single-deviation states have exactly one of the $x_i$ nonzero,
giving $r(x)=1$ and $d(x)=\nBidders$, the largest possible value.
More generally, all $k$-deviation states have $r(x)=k$ so
$d(x)=\nBidders+1-k$. This function definition gives values
directly from the representation of the state $x$, so, in
particular, the auctioneer can implement it without any knowledge
of the subspaces selected by the bidders.

The assumption of perfect search is a sufficient but not necessary
condition for the proof of Theorem~\ref{thm.NE}. The necessary
conditions are more complicated because we only need that every
single bidder deviation maps to a linear combination of infeasible
states. Thus mixing among different single-deviation states during
search (e.g., due to small eigenvalue gaps among those states), or
among states corresponding to two or more bidders deviating, does
not affect the proof.

\subsubsection{Bounded Number of Search Steps}\sectlabel{boundedSteps}

Theorem~\ref{thm.NE} shows the quantum auction has the same NE as
the classical first price auction if the search is perfect and each
bidder uses the same operator for every search step of \eq{search
update}. Since adiabatic search, run for a finite number of steps,
is not perfect we examine the effect on the NE of an imperfect
search. We show that the NE for perfect search, i.e., bidding as in
a classical first price auction and using the same operator
initially and during the search, is an $\epsilon$-equilibrium for
the auction with imperfect search. Furthermore, $\epsilon$ converges
to zero as the number of search steps goes to infinity. A strategic
profile is an $\epsilon$-equilibrium~\cite{radner80} if for every
player, the gains of unilateral defecting to another strategy is at
most $\epsilon$. This weaker equilibrium concept is useful in our
case because determining how to exploit imperfect search is
computationally difficult. Specifically, with the small eigenvalue
gaps and degeneracy it is hard to know whether imperfect search
benefits a particular bidder. Thus computational cost will likely
outweigh the small possible gain. In this situation, an
$\epsilon$-equilibrium is a useful generalization of NE.

We must prove that for any $\epsilon$ there exists an $N$ so that if
the search process uses at least $N$ steps, the equilibrium of the
game with a perfect search is also an $\epsilon$-equilibrium when
using the actual search. To do so, we bound the possible gain from
deviation based on prior knowledge of the range of possible bidder
values. That is, we assume the distribution of values has a finite
upper bound \vBound. In our context, one such bound is the maximum
bid value expressible by the announced interpretation of each
bidders qubits.

\begin{thm}
If the conditions of Theorem~\ref{thm.NE} are met, and assuming the
possible bidder values are bounded by \vBound, for any $\epsilon>0$,
there exists an $N$ so that the NE in the quantum auction with a
perfect search, shown in Theorem~\ref{thm.NE}, is also an
$\epsilon$-equilibrium of the same auction with an imperfect search
using $N$ search steps.
\end{thm}
\begin{proof}

Let $p_h$ be the probability of the highest bid wins. Let $p_{\rm
inf}$ be the probability of reaching an infeasible state. Then
$p_o=1-p_h-p_{\rm inf}$ is the probability of a bid other than the
highest bid wins.

With the adiabatic search, with nonzero eigenvalue gaps, the
probability of correctly mapping the initial to final states
converges to one as the number of search steps increases. Thus for
any $\delta>0$, there always exists a $N$ where $p_o$ is at most
$\delta$.

We define an equilibrium expected payoff function for bidder $i$
with value $v$ as $\pi^*_i(v)$, when all bidders use their
equilibrium bidding functions.

Without loss of generality, from the perspective of bidder $i$ with
value $v$, the probability of achieving the equilibrium payoff,
$\pi^*_i(v)$, if that bidder does not deviate is $1-\delta$. Thus
the expected payoff of deviating is at most $\pi^{\rm deviate}_i(v)
\leq (1-\delta)\pi^*_i(v) + \delta \vBound$ because (a) the most any
bidder can gain is bounded by \vBound, and (b) with probability
$1-\delta$ the auction either produces no profit ($p_{\rm inf}$) or
is identical to a classical auction ($p_h$).

The expected gain $g$ from deviating is the expected payoff from
deviating minus the expected payoff with no deviation, i.e.,
$g=\pi^{\rm deviate}_i(v)-\pi^*_i(v) \leq  \delta
(\vBound-\pi^*_i(v))$, which in turn is at most $\delta \vBound$.
Thus for any choice of $\delta$, there always exists an $N$ where
the maximum deviation benefit is at most $\delta \vBound$.

For any $\epsilon>0$, using $\delta=\epsilon / \vBound$ in the above
discussion shows there always exists an $N$ where the deviation is
at most $\epsilon$.
\end{proof}

\subsection{Testing for Changed Operators During
Search}\sectlabel{test change during search}

One approach to the incentive issue of changing operators during
search, described in \sect{change during search}, is for the
auctioneer to test the bidders by randomly inserting additional
probe steps in the search.

Specifically, suppose at any step of the search the auctioneer, with
some probability, decides to check a bidder by sending a new set of
qubits in a known state $\ket{\phi}$, while storing the qubits for
the search until a subsequent step. For the test step, the
auctioneer sets $D$ or $P$ to the identity operator. The state
returned by the bidder is then $U_i^\prime U_i^\dagger \ket{\phi}$
or $U_i^{\prime \dagger} U_i \ket{\phi}$, depending on which part of
the search step in \eq{search update} the auctioneer is testing.
Without loss of generality, we consider the former case.

Ideally, the bidder uses the same operator, so $U_i^\prime = U_i$
and $U_i^\prime U_i^\dagger$ is the identity. Suppose the test state
is formed from some operator $V$, randomly selected by the
auctioneer, $\ket{\phi} = V \ket{0}$. If $U_i^\prime U_i^\dagger$ is
not the identity, the returned state has the form $\alpha \ket{\phi}
+ \beta \ket{\phi_\perp}$, where $\ket{\phi_\perp}$ is some state
orthogonal to $\ket{\phi}$ and $|\alpha|^2 + |\beta|^2=1$. The
auctioneer then applies $V^\dagger$, giving
\begin{equation}
\alpha \ket{0}  + \beta \ket{a}
\end{equation}
for some value $a \neq 0$. The auctioneer then measures this
state, getting $0$ with probability $|\alpha|^2$, indicating the
bidder passes the test. Otherwise, the auctioneer observes a
different value, indicating the bidder changed the operator.

Hence the chance of getting caught depends on how often the
auctioneer checks, and how big a change the bidder makes in the
operator. Larger operator changes are more likely to be caught. This
testing behavior is appropriate as small changes are not likely to
have much affect on the search outcome, and instead simply act as an
alternate adiabatic path from initial to final states. This
technique is especially useful for risk averse bidders since then
even a small chance to be caught might be enough to prevent bidders
from wanting to change operators.

\subsection{Assigning Eigenvalues to Subspaces}\sectlabel{auctioneer}

Quantum search acts on the full space of superpositions of the
available qubits, i.e., in our case to all $2^{\nBidders \nBits}$
configurations of items and bids. In the auction context, bidders
choose operators to restrict the search to a subspace of possible
bids, namely the ones they wish to make. Conceptually, the search
described above is then restricted to the subspace selected by the
bidders.

The search can also be viewed as taking place in the full space of
$2^{\nBidders \nBits}$ configurations. The operator $U$ appearing in
the search algorithm is block diagonal (up to a permutation of the
basis states), with only the block operating on the selected
subspace relevant for the search outcome. This view of the search is
that of the auctioneer, who has no prior knowledge of the subspace
selected by each bidder. The operator $U$ is not known to any single
individual: instead its implementation is distributed among the
bidders, with each bidder implementing a part of the overall
operator. The auctioneer chooses the eigenvalues for the initial
Hamiltonian and the ordering for the qubits assigned to each bidder.
These choices, which could change during the search, affect the
incentive structure of the auction as described in \sect{behavior}.

This section describes how the auctioneer's choice of $d(x)$ can
give the same eigenvalues when restricted to the subspace actually
selected for the search. For simplicity, we suppose each bidder uses
a 2-dimensional subspace, consisting of \zero\ and the desired bid
for the single item. While not essential for the NE results
discussed above, uniformity with respect to subspace choices means
bidders are treated uniformly, so convergence of the search is
independent of the order in which the auctioneer considers the
bidders.

\subsubsection{An Example}

Consider $\nBidders=2$ bidders, each with $\nBits=2$ bits,
representing 4 values: \zero\ and three bid values $1,2,3$. A set of
2-bit operators to form a uniform superposition of the form
$(\ket{\zero}+\ket{b})/\sqrt{2}$ where $b$ is the bid value, 1, 2 or
3, is $1/\sqrt{2}$ times
\begin{displaymath}
\begin{pmatrix}
1&-1&0&0&\cr 1&1&0&0&\cr 0&0&1&-1\cr 0&0&1&1\cr
\end{pmatrix}
\begin{pmatrix}
1&0&-1&0&\cr 0&1&0&-1&\cr 1&0&1&0\cr 0&1&0&1\cr
\end{pmatrix}
\begin{pmatrix}
1&0&0&-1&\cr 0&1&1&0&\cr 0&-1&1&0\cr 1&0&0&1\cr
\end{pmatrix}
\end{displaymath}
which we can denote as $A_1,A_2,A_3$, respectively, with the first
columns giving the uniform superposition of the three possible bid
values. If the bidders select bids $b_1,b_2$, respectively, the
overall operator for the search is $U=A_{b_1} \otimes A_{b_2}$, used
in \eq{search update} to perform each step of the search. Thus in
this case there are 9 possible subspaces the two bidders can jointly
select. Up to a permutation, $U$ is block diagonal with the block
containing the nonzero entries of the first column, and hence all
the nonzero amplitude during the search, equal to
\begin{displaymath}V=
\begin{pmatrix}
1&-1&-1&1&\cr 1&1&-1&-1&\cr 1&-1&1&-1\cr 1&1&1&1\cr
\end{pmatrix}
\end{displaymath}
The search using $U$ in the full 4-bit space is thus equivalent to
one taking place in the 2-bit subspace selected by the two bidders
using this operator $V$.

The auctioneers' choice of eigenvalues, i.e., the function $d(x)$
used in \eq{search H0} should ensure the uniform superposition
within the subspace defined by the two bidders has the lowest value,
say 0, and all other eigenstates have larger values.

One possibility is the standard choice for the diagonal values
$d(x)$ when searching in the full space of $2^4$ states defined by
the $\nBidders \nBits = 4$ bits, namely the Hamming weight of each
state, i.e., the number of 1 bits in its binary representation,
ranging from 0 to 4.

An alternative approach is picking $d(x)$ so eigenvalues for the
four states appearing in $V$ have the same values as they would have
with using the Hamming weight for a 2-bit search, ranging from 0 to
2. Doing so requires selecting the eigenvalues to match the
corresponding Hamming weights for any choices the bidders make among
$A_1,A_2,A_3$. In this example, each bidder has 2 qubits, so can
represent 4 states, which we denote as $\ket{0},\ldots,\ket{3}$. The
states for both bidders are products of these individual states,
$\ket{0,0},\ldots,\ket{3,3}$. Examining the 9 possible cases for
$U$, shows a consistent set of choices is $d(\ket{x,y})$ equal to
the number of nonzero values among $x,y$. With this $d(x)$, the
adiabatic search in the subspace selected by the bidders is
identical to the standard adiabatic search for two bits. This choice
treats both bidders identically.

In this case we see the auctioneer can arrange the adiabatic search
to operate symmetrically no matter what choice of subspace each
bidder makes (i.e., no matter what value each bidder decides to
bid). Thus from the point of view of the bidders, the search, in
effect, takes place within the subspace of possible values defined
by their bid selections.

\subsubsection{General Case}

For arbitrary numbers of bidders $\nBidders$ and bits $\nBits$, we
consider a single-item auction so each bidder would, ideally, pick
an operator giving just two terms, with $b^{(j)}$ the bid of bidder
$j$ for the single item and no bits needed to specify which item the
bidder is interested in. The choice of $b^{(j)}$ corresponds to the
bidder picking a 2-dimensional subspace of the $2^\nBits$ possible
states. The product of these subspaces gives a subspace $S$ of all
$\nBidders \nBits$ qubits used in the auction. The subspace $S$ has
dimension $2^\nBidders$ and its states $x_S$ can be viewed as
strings of $\nBidders$ bits. More specifically, we suppose bidder
$j$ implements the operator $U_j$ such that the rows and columns
corresponding to \zero\ and $b^{(j)}$ have nonzero values only for
positions \zero\ and $b^{(j)}$. That is, the elements of $U_j$ for
these two values form a $2\times 2$ unitary matrix.

If the auctioneer knew the subspace $S$, the eigenvalue function
$d(x)$ used in \eq{search H0} could be selected to match any desired
function $d_S(x_S)$ of the states in $x_S \in S$. Without such
knowledge, this is possible only for some choices for $d_S$.

\begin{thm}
Provided $d_S(x_S)$ depends only on the Hamming weight of the states
$x_S$, a single choice of $d(x)$ in the full space corresponds to
$d_S(x_S)$ in all possible subspaces the bidders could select that
include the null set.
\end{thm}

\begin{proof}
Consider the full operator $U$ given by \eq{U}. For the element
$U_{x,y}$, express the $\nBidders \nBits$ bits defining the states
$x$ and $y$ as sequences of $\nBits$-bit values,
$x_1,\ldots,x_\nBidders$ and $y_1,\ldots,y_\nBidders$, respectively,
with each $x_i$ and $y_i$ between 0 and $2^\nBits-1$. From \eq{U},
\begin{displaymath}
U_{x,y} = \prod_{i=1}^\nBidders (U_i)_{x_i,y_i}
\end{displaymath}
The matrix $U$ is of size $2^{\nBidders \nBits} \times
2^{\nBidders \nBits}$ while each $U_i$ is of size $2^{\nBits}
\times 2^{\nBits}$.

Consider the first column of $U$, i.e., $y=0$. $U_{x,0}$ is nonzero
only for those $x$ such that all the $(U_i)_{x_i,0}$ are nonzero.
For this to be the case, each $x_i$ is either 0 (corresponding to
$\ket{\zero}$ for that bidder's superposition) or $x_i = b^{(i)}$,
i.e., the bid value. Similarly, for all columns with each $y_i$
equal to 0 or $b^{(i)}$. These values for $x,y$ are precisely the
states in the selected subspace of the bidders, $S$. For these
choices of $x_i,y_i$, we can map 0 (i.e., $\nBits$ bits all equal to
zero) to the single bit 0, and each $b^{(i)}$ (specified by values
for $\nBits$ bits) to the single bit 1. This establishes a
one-to-one mapping from states in the full space, of $\nBidders
\nBits$ bits corresponding to the product of  bidders'
superpositions, to states in the subspace treated as $\nBidders$-bit
vectors. Thus a function $d_S(x_S)$ applied to the subspace that
depends on the Hamming weight, i.e., the number of 1 bits in $x_S$,
is the same as a function on the full space depending on the number
of nonzero $x_i$ values in $x = x_1,\ldots,x_\nBidders$.

We must show that a single choice of function $d(x)$ in the full
space maps to the desired $d_S(x_S)$ in {\em any} choice of bidder
subspaces. To see this is the case, consider any state in the full
space $x = x_1,\ldots,x_\nBidders$. Among these $x_i$, suppose $h$
are nonzero, denoted by $x_{a_1},\ldots,x_{a_h}$. This state $x$
will appear in all selected subspaces in which bidder $a_j$ bids
$b^{(a_j)} = x_{a_j}$, for $j=1,\ldots,h$, and the remaining bidders
have any choice of bid. That is, $x$ appears in
$(2^\nBits-1)^{\nBidders-h}$ possible subspaces $S$. Since $x$ has
exactly $h$ nonzero values, in each of these possible subspaces it
maps to a state $x_S$ with exactly $h$ bits equal to 1, i.e., it has
the same Hamming weight, $h$, in all possible subspaces in which it
appears. Thus any choice of function $d_S(x_S)$ depending only on
the Hamming weight of $x_S$ will have the same value in all these
possible subspaces. This observation allows the auctioneer to select
that common value as the value for $d(x)$, consistently giving the
desired eigenvalue function for any possible subspace. Since this
holds for all values of $h$, the auctioneer can operate in the full
space with identical search behavior no matter what subspace the
bidders select.
\end{proof}

For the auctioneer to operate without knowledge of the actual
subspace selected by the bidders and treat bidders identically, we
need $d(x)$ to map to the same function on any subspace selected. In
this case, the search proceeds exactly as if the auctioneer did know
the subspace choices made by the bidders. The theorem gives one type
of function for in which this is the case. In particular, \eq{d(x)
example} is an example of a function satisfying this theorem.

\section{Multiple Items and Combinatorial Auction}\sectlabel{Multiple}

While the paper focuses on the single item first-price sealed-bid
auction, the quantum protocol can apply to multiple items by
changing the interpretation of the bids, i.e., the bidding language.
Such changes affect the counting of deviation and feasible states,
so we must check the validity of Theorem~\ref{thm.NE}.

In the single item case, each bidder uses the \nBits\ qubits to
specify the bid amount. With multiple items, the bid must specify
both the items of interest and a bid amount for the items. Various
bidding languages can encode this information.

For multiple items, we divide the $\nBits$\ qubits allocated to each
bidder into two parts: $\nBitsi$ bits to denote a bundle of items
and $\nBitsp$ bits to denote bid value (so $\nBits = \nBitsi +
\nBitsp$). Since qubits are expensive, a succinct representation of
items is best. Depending on the type of auction, we have various
choices with different efficiency in using bits. For example, the
$\nBitsi$ item bits could indicate the item in the bid, allowing
$\nBitsi$ qubits to specify up to $2^\nBitsi$ different items.
Another case is multiple units of a single item, so $\nBitsi$ could
specify how many units a bidder wants (with the understanding the
bid is for all those units not a partial amount) so the bits could
specify $2^\nBitsi$ different numbers. In the general case, bids are
on arbitrary sets of items or \emph{bundles}, and we represent a
bundle with \nItems\ bits, 1 if the corresponding item is a part of
the bundle and 0 otherwise, i.e., $\nItems =\nBitsi$. We focus on
this general case in the remainder of the section. Allowing bids on
sets of items is called a combinatorial auction~\cite{cramton06}.

With multiple items, the bid operator $\psi_{j} = U_{j} \psiInit$
gives a superposition of bids of the form $\ket{I_{i}^{(j)},
b_{i}^{(j)}}$ where $b_{i}^{(j)}$ is bidder $j$'s bid for a bundle
of items $I_{i}^{(j)}$. In this notation, the null bid is
$\ket{\zero,b}$, and the specified amount $b$ is irrelevant so we
take it to be zero in the examples. A superposition specifies a set
of distinct bids, with at most one allowed to win.

\begin{eg}Consider a combinatorial auction with two
items $X$, $Y$ and integer prices ranging from 0 to 3. With $\nBits
= 4$ bits for each bidder, using 2 bits each to specify item bundles
and prices, is sufficient to specify the bids. The full space for a
bidder has dimension $2^\nBits = 16$, consisting of $4$ possible
item bundle choices and $4$ price choices. Suppose a bidder places a
bid
\begin{displaymath}
\frac{1}{\sqrt{3}}(\ket{\zero,0} + \ket{X,1}+ \ket{(X,Y),2})
\end{displaymath}
i.e., a bid of 1 for item X alone, and 2 for the bundle of both
items. In this case, the bidder is not interested in item $Y$ by
itself. The dimension of the subspace of this bid is 3. Another
example is the bid
\begin{displaymath}
\frac{1}{\sqrt{4}}(\ket{\zero,0} + \ket{X,1}+ \ket{X,3} +
\ket{(X,Y),4})
\end{displaymath}
The dimension of the subspace is 4. This superposition has multiple
bids on the same item $X$.
\end{eg}

This bidding language is both expressive and compact. For instance,
a superposition of bundles of items readily expresses exclusive-or
preferences, where a bidder wants at most one of the bundles. It is
also compact because superpositions allow the bidder to use exactly
the same qubits to place no bid (i.e., \zero) and to place all the
exponential number of bundles in a combinatorial auction.

An allocation, as defined in \sect{bids}, is a list of bids, one
from each bidder. With multiple items, an allocation is feasible if
the item sets are pairwise disjoint. As in the single item case, we
consider the allocation when all item sets are empty as infeasible.
The value of a feasible allocation is the sum total of the bid
values of the different bids in the allocation. The number of
feasible states is $((\nBidders + 1) ^ \nItems -1)2^{\nBidders
\nBitsp}$. This is because we can assign \nItems\ items among
\nBidders\ bidders where all items need not be allocated in
$(\nBidders + 1) ^ \nItems$ ways. The factor $\nBidders + 1$ allows
for some items to remain unallocated. Since the allocation when all
bidders place the null bid is an infeasible state, we subtract 1.
Each bidder can specify $2^\nBitsp$ different prices for the bundle
giving $2^{\nBidders \nBitsp}$ possible choices for \nBidders\
bidders. Note that the number of feasible states for a single item,
$\nItems = 1$, is different from that in \sect{behavior} because
here we have changed the bidding language to represent items also.

The null bid in our protocol simplifies the evaluation of
allocations for combinatorial auctions. To see this, consider a
protocol without the null bid. In a single item case, $F(x)$ for any
allocation vector $x$ would be maximum of the bids placed by the
different bidders on the item, which is fairly easy to compute. But
in the case of multiple items, there could be several allocations
for a vector $x$. For example suppose Alice bids on the set
$\{A,B\}$ and Bob bids on $\{B,C\}$. Without the null set then both
bids appear in the same state and have to be evaluated by $F(x)$.
The possible allocation to the bidders are
\begin{enumerate}
\item none to either
\item $\{A,B\}$ to Alice
\item $\{B,C\}$ to Bob, and
\item $\{A,B\}$ to Alice and $\{B,C\}$ to Bob (which is
infeasible)
\end{enumerate}
$F(x)$ will have to compute the maximum of the values in all these
states. This is computationally complex when there are many items.
By contrast, the bidding language with the null bid avoids this
combinatorial evaluation within the search function $F(x)$.

As in the case of single item auctions, we restrict ourselves to a
one-shot sealed bid classical combinatorial auction that we
implement in a quantum setting. The total number of states is
$2^{\nBits \nBidders}$ and the total number of single bidder
deviations states is $\nBidders(2^{\nBits} -1)$. These expressions
are the same as the single item case. The condition for all
single-deviation states to be mapped to infeasible allocations,
resulting in no winner, is
\begin{equation}\eqlabel{multiItemCondition}
2^{\nBidders \nBits} - ((\nBidders + 1) ^ \nItems -1)2^{\nBidders
\nBitsp } \geq \nBidders (2^{\nBits} - 1)
\end{equation}
This condition holds for cases relevant for auctions as seen in the
following claim.
\begin{clm} \label{multiClaim}
\eq{multiItemCondition} is true for all integers $\nItems, \nBitsp
\geq 1$ and $\nBidders \geq 2$.
\end{clm}
\begin{proof}
Recall $\nBits = \nItems + \nBitsp$. We prove a stronger condition
for integers $\nBidders, \nItems \geq 2$, i.e., there exists enough
infeasible states to handle joint deviations up to $\nBidders -1$
bidders. The number of $k$-bidder deviation states is the same as
the single-item case, i.e., \eq{k deviation states}. Thus this
stronger condition, with the same right-hand side as \eq{condition
with collusion}, is
\begin{equation}\eqlabel{multiItemCollusion}
2^{\nBidders \nBits} - ((\nBidders + 1) ^ \nItems -1)2^{\nBidders
\nBitsp } \geq  2^{\nBidders \nBits} - 1 - (2^{\nBits} -1)^\nBidders
\end{equation}
Hence \eq{multiItemCollusion} is true if
\begin{eqnarray*}
(2^{\nBits} -1)^\nBidders & \geq & ((\nBidders + 1) ^ \nItems
-1)2^{\nBidders \nBitsp } \\
\Leftrightarrow 2^{\nBitsp \nBidders}(2^{\nItems} -
2^{-\nBitsp})^\nBidders & \geq & ((\nBidders + 1) ^ \nItems
-1)2^{\nBidders \nBitsp } \\
\Leftrightarrow (2^{\nItems} - 2^{-\nBitsp})^\nBidders & \geq &
(\nBidders + 1)^ \nItems -1
\end{eqnarray*}
Since $2^{-\nBitsp} \leq 1$, \eq{multiItemCollusion} is true if
$$(2^{\nItems} - 1)^\nBidders \geq (\nBidders + 1)^ \nItems -1$$
which is true if
$$(2^{\nItems} - 1)^{\frac{1}{\nItems}} \geq (\nBidders + 1) ^
{\frac{1}{\nBidders }}$$

\noindent Let $f(\nItems) \equiv (2^{\nItems} -
1)^{\frac{1}{\nItems}}$ and $g(\nBidders) \equiv (\nBidders + 1) ^
{\frac{1}{\nBidders }}$. We establish the required inequality,
$f(\nItems) \geq g(\nBidders)$, by showing $f(\nItems)$ is
increasing in $\nItems$ when $\nItems \geq 2$, $g(\nBidders)$ is
decreasing in $\nBidders$ when $\nBidders \geq 2$ and noting $f(2) =
g(2) = \sqrt{3}$.

\noindent Taking the derivative of $f(\nItems)$ with respect to
$\nItems$, we get,
$$\frac{(2^{\nItems} - 1)^{\frac{1}{\nItems}}}{\nItems} \left(
\frac{2^\nItems \ln(2)}{2^{\nItems} - 1} - \frac{\ln(2^{\nItems} -
1)}{m}\right)$$ This is positive if and only if
$$ \frac{2^\nItems}{2^\nItems -1} \frac{\nItems}{\log_2(2^\nItems - 1)}
> 1$$ This is true because $\log_2(2^\nItems - 1) < \log_2
(2^{\nItems})  = \nItems$ and hence both fractions in the expression
are greater than 1. Thus, $f(\nItems)$ is increasing for all
$\nItems \geq 2$.

\noindent Taking derivative of $g(\nBidders)$ with respect to
$\nBidders$, we get,
$$ \frac{(\nBidders + 1)^{\frac{1}{\nBidders }}}{\nBidders}\left(\frac{1}{1+\nBidders} - \frac{\ln(1+\nBidders)}{\nBidders}\right)$$
This is negative if and only if
$$ \ln(1+\nBidders) \frac{1+\nBidders}{\nBidders}
> 1$$ This is true for $\nBidders \geq$ 2. Thus $g(\nBidders)$ is
decreasing in $\nBidders$ for $\nBidders \geq 2$.

Thus we have shown that \eq{multiItemCondition} is true for
$\nBidders, \nItems \geq 2$. It can be easily checked that
\eq{multiItemCondition}, is not true for $\nBidders = 1$ and true
when $\nBidders = 2$ and $\nItems =1$.
\end{proof}

Thus, if a classical combinatorial auction has a NE then the
corresponding quantum auction protocol also has a NE with respect to
initial state deviations. Also there is an $\epsilon$-equilibrium of
the same auction with an imperfect search using $N$ search steps.
Moreover, the stronger condition of \eq{multiItemCollusion} shows
that in auctions with at least two bidders ($\nBidders
>1$), there are enough infeasible states to give no winner
for any deviation of initial amplitudes that up to $\nBidders -1$
can produce. Thus no groups, up to size $\nBidders -1$, can collude
to benefit from initial amplitude deviations in the quantum auction.

\section{Applications of Quantum Auctions}\sectlabel{applications}

Two properties of quantum information may provide benefits to
auctioneers and bidders: the ability to compactly express
complicated combinations of preferences via superpositions and
entanglement and the destruction of the quantum state upon
measurement. This section describes some economic scenarios that
could benefit from these properties.

As one economic application, quantum auctions provide a natural way
to solve the \emph{allocative externality}
problem~\cite{jehiel05,ranger05}. In this situation, a bidder's
value for an item depends on the items received by other bidders.
For example, consider companies bidding on a big government project
requiring multiple companies to work on different parts. Allocative
externality refers to the issue that the costs for a company which
wins a contract for one part depends on which other companies win
other parts. So company A may be willing to bid more aggressively if
it knows that company B will work on related parts. Multiple
simultaneous auctions for separate parts will not handle these
interdependencies and thus will be inefficient. One possible
solution is to let companies form partnership bids. That is joint
bids that are accepted together or not at all. Quantum information
processing allows for a natural way of forming partnership bids via
entanglement. With the protocol described in \sect{Multiple},
multiple bids can be entangled so they will either all be accepted
together or none will be. Furthermore, quantum auctions may provide
more flexibility with respect to information privacy of partnership
bids than classical methods.

Specifically, with multiple items, groups of bidders could select
joint operators on their combined qubits, allowing them to express
joint constraints (e.g., where they either all win their specified
items or none of them do) without any of the other bidders or
auctioneer knowing this choice. The bidders do so by creating an
entangled state instead of the factored form for their qubits.  Thus
employing quantum entanglement provides bidders a natural way for
expressing any allocative externality. This possibility shows
bidding languages based on qubits are highly expressive and compact
because bidders can use the same bits to express their individual
bids and joint bids via entanglement.

\begin{eg}
Alice and Bob could jointly form the state
\begin{equation}
\frac{1}{\sqrt{3}} ( \ket{\zero,0,\zero,0} +
\ket{I_A,b_{A},I_B,b_{B}} + \ket{I_C,b_{C},I_D,b_{D}} )
\end{equation}
to represent the bidders willing to pay $b_{A}$ and $b_{B}$ for
items $I_A$ and $I_B$, or to pay $b_{C}$ and $b_{D}$ for items $I_C$
and $I_D$, but they are not willing to buy other combinations, such
as $I_A$ for Alice and $I_D$ for Bob.
\end{eg}
In this scenario, a direct representation of bids, i.e., without a
null bid, would not guarantee the joint preferences are satisfied
for all entangled bidders or none of them. That is, without null
bids, the superposition could not express the joint preference
through entanglement.

A group of $k$ bidders operating jointly on their qubits to form
entangled bids could also produce initial amplitudes involving up to
$k$-bidder deviation states. However the discussion with
\eq{multiItemCollusion} on multiple item auctions shows our protocol
can handle all deviation states a group of up to $\nBidders-1$
bidders can produce, i.e., by mapping them to infeasible outcomes.
Thus the additional expressivity used for joint bids does not
introduce additional opportunities for collusion to change the
outcomes via initial amplitude selection.

A second economic application for quantum auctions arises from their
privacy guarantee for losing bids. This property is economically
useful when bidders have incentives to hide information. An example
is a scenario in which companies are bidding for government
contracts year after year. A company's bid usually contains
information about its cost structure. If there is reasonable
expectation that the losing bids will be revealed, a company may
want to bid less aggressively to reduce the amount of information
passed to its competition for use in future auctions. This will lead
to a less efficient auction than if bidders reveal their true
values. In this situation, a privacy guarantee on the losing bids
enables bidders to bid with less inhibition. More generally, this
privacy issue is only relevant when there are additional
interactions between these companies after the auction is concluded,
such as future auctions or negotiations where participants may be at
a disadvantage if their values are known to others.

This strong privacy property is unique to quantum information
processing. Privacy can be enforced via cryptographic methods for
multi-player computation~\cite{goldreich98}, and in an auction can
keep losing bids secret~\cite{naor99}. However, the information on
the bids, and the key to decrypt them, remains after the auction
completes. People who have access to the key may be legally
compelled to reveal the information or choose to sell it. So while
cryptography can be secured computationally, it cannot guarantee the
integrity of the person(s) who have the means to decrypt the
information. On the other hand, the quantum method destroys losing
bids during the search for the winning one and it is physically
impossible to reconstruct the bids after the auction process.
Similarly, some of the other properties of quantum auctions, such as
correlations for partnership bids, can be provided
classically~\cite{meyer04}. Moreover, quantum mechanisms are readily
simulated classically~\cite{vanenk02} (as long as they involve at
most 20 to 30 qubits). However, these classical approaches lack the
information security of quantum states. More study is needed to
determine scenarios where the privacy property of the quantum
protocol is significant.

\section{Discussion}\sectlabel{discussion}

This paper describes a quantum protocol for auctions, gives a game
theory analysis of some strategic issues the protocol raises and
suggests economic scenarios that could benefit from these auctions.
These include the privacy of bids and the possibility of addressing
allocative externalities. The search used in our protocol can use
arbitrary criteria for evaluating allocations, thereby implementing
other types of auctions with quantum states. Thus while we focus our
attention on the first-price sealed-bid auction, the protocol is
more general: it can implement other pricing and allocation rules,
as well as multiple-unit-multiple-item auctions with combinatorial
bids. For example we can use this protocol in a multiple stage,
iterative auction. In fact, the protocol supports general bidding
languages.

Encoding bids in quantum states raises new game theory issues
because the bidders' strategic choices include specifying amplitudes
in the quantum states. The auction is not only probabilistic, but
the winning probability is not just a function of the amount bid.
Instead a bidder can change the probability of winning by altering
the amplitudes of the quantum states encoding his bid. For example,
in the context of the first-price sealed-bid auction, the auction
does not guarantee the allocation of the item to the highest bidder.

We show that the correct design of the protocol can solve a specific
version of this incentive problem. The salient design feature is an
incentive compatible mechanism so that bidders do not want to cheat,
as opposed to an algorithmic secure protocol that prevents bidders
from cheating. Thus, our design is an example of a quantum
algorithm, in this case adiabatic search, tuned to improve incentive
issues rather than the usual focus in quantum information processing
on computation or security properties of algorithms.

In addition, we show that the Nash equilibrium of the corresponding
classical first-price sealed-bid auction is an
$\epsilon$-equilibrium of the quantum auction and that $\epsilon$
converges to zero when the quantum search associated with the
protocol uses an increasing number of steps, under the conditions
listed in Theorem~\ref{thm.NE}. This result is with respect to
changes in the initial state of the search. It remains to be seen
whether other bidder strategies give some unilateral benefit,
requiring further adjustments to the auction design.

There are multiple directions for future work. First, we plan a
series of human subject experiments on whether people can indeed bid
effectively in the simple quantum auction scenario described in this
paper. As with previous experiments with a quantum public goods
mechanism~\cite{chen06}, such experiments are useful tests of the
applicability of game theory in practice, and also suggest useful
training and decision support tools. In particular, people's
behavior in a quantum auction could differ from game theory
predictions that people select a Nash equilibrium based on idealized
assumptions of human rationality and full ability to evaluate
consequences of strategic choices with uncertainty.

Second, we plan to extend studies of quantum auctions to more
complicated economic scenarios, such as one with allocative
externality. Our analysis considers a single auction. An interesting
extension is to a series of auctions for similar items. If auctions
are repeated, the game theory analysis is more
complicated~\cite{wilson92}. In particular, privacy concerns become
more significant since information revealed by a bidder's behavior
in one auction may benefit other bidders in later auctions.

The quantum auction destroys all information about the losing bids.
As a result, it is not possible to conduct after-the-fact audits to
verify that the auction has been conducted correctly. Is there a way
to modify the mechanism to enable audits while preserving some of
the privacy guarantees? Security is another interesting issue. For
example, there may be third parties, aside from the auctioneer and
bidders who are interested in intercepting and changing bits in
transit. Auctioneers may have incentives to detect a bidder's bid or
skew auction results. The question is whether we can build security
around the protocol to prevent or at least detect these types of
attacks.

Similarly, many economics issues surrounding the protocol remain to
be resolved. For example, people behave as if they are risk averse
in auction situations~\cite{cox88,chen98} which can change the
predictions of game theory. Another issue arises from the
possibility of multiple Nash equilibria. We have only shown that the
desirable outcome is \emph{an} equilibrium. The quantum protocol can
also have other equilibria. Since the Nash equilibrium concept alone
does not indicate how people select one equilibrium over another,
additional study is needed to determine when the desirable outcome
is likely to occur.

Our protocol makes only limited use of quantum states, in
particular encoding bids in the subspace selected by the bidders
but not using the amplitudes separately. Thus it would be
interesting to examine extensions to the protocol exploiting the
wider range of options for bidders. For example, a protocol might
use amplitudes of superpositions to indicate a bidder's
probabilistic preferences, say, as in constructing a portfolio of
items with various expected values and risks. Such portfolios
could be useful if bidders have some uncertainty in their values
(e.g., in bidding for oil field exploration rights) rather than
the standard private value framework considered in this paper,
where bidders know their own values for the items. With uncertain
values, probabilistic bids could allow bidders to match their risk
preferences along with their value estimates within the auction
process.

As a final note, the number of qubits necessary to conduct an
auction is small compared to the requirement of complex
computations such as factoring. For example, if each bidder uses 7
bits (corresponding to $2^7$ or about 100 bid values) and there
are 3 bidders, about 25 qubits are needed, considerably less than
thousands needed for factoring interesting-sized numbers. Thus
with the advancement of quantum information processing
technologies, economics mechanisms could be early feasible
applications.

\small
\section*{Acknowledgments}
We have benefited from discussions with Raymond Beausoleil, Saikat
Guha, Philip Kuekes, Andrew Landahl and Tim Spiller. This work was
supported by DARPA funding via the Army Research Office contract
\#W911NF0530002 to Dr. Beausoleil. This paper does not necessarily reflect the position or the
policy of the Government funding agencies, and no official
endorsement of the views contained herein by the funding agencies
should be inferred.


\end{document}